\documentclass{fac}

\usepackage{graphicx}
\usepackage{amssymb}
\usepackage{amsfonts}
\usepackage{amsmath}
\usepackage{dsfont}
\usepackage{MnSymbol}
\usepackage{mathtools}
\usepackage{epsfig}
\usepackage{epstopdf}

\newtheorem{theorem}{Theorem}[section]
\newtheorem{definition}[theorem]{Definition}
\newtheorem{proposition}[theorem]{Proposition}

\title[Draft of Communicating Concurrent Processes]
      {Communicating Concurrent Processes}

\author[Yong Wang]
    {Yong Wang\\
     College of Computer Science and Technology,\\
     Faculty of Information Technology,\\
     Beijing University of Technology, Beijing, China\\
     }

\correspond{Yong Wang, Pingleyuan 100, Chaoyang District, Beijing, China.
            e-mail: wangy@bjut.edu.cn}

\pubyear{2017}
\pagerange{\pageref{firstpage}--\pageref{lastpage}}

\begin{document}
\label{firstpage}

\makecorrespond

\maketitle

\begin{abstract}
Process algebra CSP only permits a process to engage in one event on a moment and records this single event into the traces of the process. CSP cannot process events simultaneously, it treat the events occurred simultaneously as one single event. We modify CSP to process the events occurred simultaneously, which is called communicating concurrent processes (CCP).
\end{abstract}


\section{Introduction}\label{int}

Process algebras are well-known formal theories to capture computational concepts in computer science, especially parallelism and concurrency. CCS \cite{CCS}, CSP \cite{CSP} and ACP \cite{ACP} are three process algebras, and CSP is used widely in verify the behaviors of computer systems for its rich expressive power.

CSP only permits a process to engage in one event on a moment and records this single event into the traces of the process. CSP cannot process events simultaneously, it treat the events occurred simultaneously as one single event. In this paper, we modify CSP to process the events occurred simultaneously, this work is called communicating concurrent processes (CCP). The main difference between CCP and CSP is the treatment of concurrency.

This paper is organized as follows. In section \ref{pct}, we introduce the concepts and laws of processes and concurrent traces in CCP. In section \ref{concurrency}, we give the concepts and laws of concurrency in CCP. we discuss the rest contents of CCP corresponding to CSP briefly, for they are almost same. And we conclude this paper in section \ref{con}.

\section{Processes and Concurrent Traces}\label{pct}

In this section, we modify processes and traces to be suitable for CCP. Some concepts are modified for this situation, the others are just retype from CSP, including concepts of processes and traces.

\subsection{Processes}

The intuitions of the concept events, event names, alphabet, processes, are the same as CSP. We also use the following conventions:

\begin{itemize}
  \item Words in lower-case letters denote distinct events;
  \item Words in upper-case letters denote specific defined processes;
  \item The letters $x,y,z$ are variables denoting events;
  \item The letters $A,B,C$ denote sets of events;
  \item The letters $X,Y$ are variables denoting processes;
  \item The alphabet of process $P$ is denoted $\alpha P$;
  \item The process with alphabet $A$ which never actually engages in any of the events of $A$ is called $STOP_A$;
  \item The process with alphabet $A$ which can engage in any event of $A$ is called $RUN_A$.
\end{itemize}

\begin{definition}[Prefix]
Let $x_1,\cdots,x_n$ ($n\in\mathbb{N}$) be events and let $P$ be a process. Then the prefix

\[(\{x_1,\cdots,x_n\}\rightarrow P)\]

denotes a process engaged in the events $x_1,\cdots,x_n$ simultaneously firstly, and then behaving as $P$, with $\alpha(\{x_1,\cdots,x_n\}\rightarrow P)=\alpha P$, if $x_1,\cdots, x_n\in \alpha P$.
\end{definition}

To describe infinite behaviors, we also introduce recursion.

\begin{definition}[Recursion]
For a prefix guarded expression $F(X)$ containing the the process name $X$, with the alphabet $A$ of $X$, then the recursive equation

\[X=F(X)\]

has a unique solution denoted as $\mu X:A\bullet F(X)$.
\end{definition}

\begin{definition}[Choice]
If $\{x_1,\cdots, x_n\}\neq\{y_1,\cdots,y_m\}$ with $m,n\in\mathbb{N}$, then the choice

\[(\{x_1,\cdots,x_n\}\rightarrow P \mid \{y_1,\cdots,y_m\}\rightarrow Q)\]

will initially engage in either of the events $\{x_1,\cdots,x_n\}$ or $\{y_1,\cdots, y_m\}$ simultaneously, with $\alpha(\{x_1,\cdots,x_n\}\rightarrow P \mid \{y_1,\cdots,y_m\}\rightarrow Q)=\alpha P = \alpha Q$.
\end{definition}

\begin{proposition}[Laws of Processes]\label{LoP}
We have the following laws of processes.

\begin{itemize}
  \item L1. $(\{x_1,\cdots,x_n\}:A\rightarrow P(x_1,\cdots,x_n))=(\{y_1,\cdots,y_m\}: B\rightarrow Q(y_1,\cdots,y_m))\equiv(m=n\wedge A=B\wedge\forall x_i:A\bullet P(x_i)=Q(x_i))$
      \begin{itemize}
        \item L1A. $STOP\neq(\{d_1,\cdots, d_n\}\rightarrow P)$
        \item L1B. $(\{c_1,\cdots,c_n\}\rightarrow P)\neq(\{d_1,\cdots,d_m\}\rightarrow Q)$ if $\{c_1,\cdots,c_n\}\neq\{d_1,\cdots,d_m\}$
        \item L1C. $(\{c_1,\cdots,c_n\}\rightarrow P\mid\{d_1,\cdots,d_m\}\rightarrow Q)=(\{d_1,\cdots,d_m\}\rightarrow Q\mid \{c_1,\cdots,c_n\}\rightarrow P)$
        \item L1D. $(\{c_1,\cdots,c_n\}\rightarrow P)=(\{c_1,\cdots,c_n\}\rightarrow Q)\equiv P=Q$
      \end{itemize}
  \item L2. If $F(X)$ is a guarded expression, $(Y=F(Y))\equiv(Y=\mu X\bullet F(X))$.
  \begin{itemize}
    \item L2A. $\mu X\bullet F(X)=F(\mu X\bullet F(X))$
  \end{itemize}
\end{itemize}
\end{proposition}

\begin{proof}
These laws can be proven straightforwardly from the related definitions of processes.
\end{proof}

The processes also can be implemented by LISP as in CSP.

\subsection{Concurrent Traces}

\begin{definition}[Concurrent Traces]
A trace is a sequence of symbols, separated by commas and enclosed in angular brackets. A concurrent trace is a trace within a pair of commas, there may be a set of symbols.

\begin{itemize}
  \item $\langle\{x_1,\cdots,x_n\},\{y_1,\cdots,y_m\}\rangle$ is a concurrent trace contains two event sets, $\{x_1,\cdots,x_n\}$ followed by $\{y_1,\cdots,y_m\}$;
  \item $\langle\rangle$ is the empty trace.
\end{itemize}

In the following, we also call a concurrent trace as a trace. And we write $s,t,u$ for traces, $S,T,U$ for sets of traces, and $f,g,h$ for functions.
\end{definition}

\begin{definition}[Catenation]
The catenation of a pair of traces $s$ and $t$ is defined as $s\smallfrown t = \langle s, t\rangle$.
\end{definition}

\begin{proposition}[Laws of Catenation]
The laws of catenation for concurrent traces are the same for traces in CSP, we retype them as follows.
\begin{itemize}
  \item L1. $s\smallfrown\langle\rangle = \langle\rangle\smallfrown s = s$
  \item L2. $s\smallfrown(t\smallfrown u) = (s\smallfrown t)\smallfrown u$
  \item L3. $s\smallfrown t = s\smallfrown u\equiv t = u$
  \item L4. $s\smallfrown t = u\smallfrown t\equiv s=u$
  \item L5. $s\smallfrown t=\langle\rangle\equiv s=\langle\rangle\wedge t=\langle\rangle$
  \item L6. For $t^n$ is $n$ copies of $t$ catenated with each other, $t^0=\langle\rangle$
  \item L7. $t^{n+1}=t\smallfrown t^n$
  \item L8. $t^{n+1}=t^n\smallfrown t$
  \item L9. $(s\smallfrown t)^{n+1}=s\smallfrown(t\smallfrown s)^n\smallfrown t$
\end{itemize}
\end{proposition}

\begin{proof}
These laws can be proven straightforwardly from the related definition of catenation.
\end{proof}

\begin{definition}[Restriction]
The restriction $(t\upharpoonright A)$ of the trace $t$ restricted to the set of symbols $A$, is obtained by omitting all symbols outside $A$ from $t$.
\end{definition}

\begin{proposition}[Laws of Restriction]
The laws of restriction for concurrent traces are the same as for traces in CSP, we also retype them as follows.
\begin{itemize}
  \item L1. $\langle\rangle\upharpoonright A=\langle\rangle$
  \item L2. $(s\smallfrown t)\upharpoonright A=(s\upharpoonright A)\smallfrown (t\upharpoonright A)$
  \item L3. $\langle \{x_1,\cdots,x_n\}\rangle\upharpoonright A = \langle\{x_1,\cdots,x_n\}$, if $x_1,\cdots,x_n\in A$
  \item L4. $\langle\{y_1,\cdots, y_m\}\rangle\upharpoonright A=\langle\rangle$, if $y_1,\cdots, y_m\notin A$
  \item L5. $s\upharpoonright\{\}=\langle\rangle$
  \item L6. $(s\upharpoonright A)\upharpoonright B = s\upharpoonright(A\cup B)$
\end{itemize}
\end{proposition}

\begin{proof}
These laws can be proven straightforwardly from the related definition of restriction.
\end{proof}

\begin{definition}[Head and Tail]
For a nonempty trace $s$, the head of $s$ is the first sequence denoted $s_0$, the left sequence is the tail of $s$ denoted $s'$.
\end{definition}

\begin{proposition}[Laws of Head and Tail]
The laws of head and tail for concurrent traces are as follows.
\begin{itemize}
  \item L1. $(\langle\{x_1,\cdots,x_n\}\rangle\smallfrown s)_0=\{x_1,\cdots,x_n\}$
  \item L2. $(\langle\{x_1,\cdots,x_n\}\rangle\smallfrown s)'=s$
  \item L3. $s=(\langle s_0\rangle\smallfrown s')$, if $s\neq\langle\rangle$
  \item L4. $s=t\equiv(s=t=\langle\rangle\vee(s_0=t_0\wedge s'=t'))$
\end{itemize}
\end{proposition}

\begin{proof}
These laws can be proven straightforwardly from the related definition of head and tail.
\end{proof}

\begin{definition}[Star]
The star $A^\ast$ of $A$ is defined as follows. $A^\ast=\{s|s\upharpoonright A = s\}$
\end{definition}

\begin{proposition}[Laws of Star]
The laws of star for concurrent traces are as follows.
\begin{itemize}
  \item L1. $\langle\rangle\in A^\ast$
  \item L2. $\langle\{x_1,\cdots,x_n\}\rangle\in A^\ast\equiv\{x_1,\cdots,x_n\}\in A$
  \item L3. $(s\smallfrown t)\in A^\ast\equiv s\in A^\ast\wedge t\in A^\ast$
  \item L4. $A^\ast=\{t|t=\langle\rangle\vee(t_0\in A\wedge t'\in A^\ast)\}$
\end{itemize}
\end{proposition}

\begin{proof}
These laws can be proven straightforwardly from the related definition of star.
\end{proof}

\begin{definition}[Ordering]
For traces $s,u,t$, such that $s\smallfrown u = t$, The ordering relation of concurrent traces is defined as follows. $s\leq t = (\exists u\bullet s\smallfrown u = t)$.
\end{definition}

\begin{proposition}[Laws of Ordering]
The laws of the ordering relation for concurrent traces are the same for traces in CSP, we retype them as follows.
\begin{itemize}
  \item L1. $\langle\rangle\leq s$
  \item L2. $s\leq s$
  \item L3. $s\leq t\wedge t\leq s\Rightarrow s=t$
  \item L4. $s\leq t\wedge t\leq u\Rightarrow s\leq u$
  \item L5. $(\langle\{x_1,\cdots,x_n\}\rangle\smallfrown s)\leq t\equiv t\neq\langle\rangle \wedge \{x_1,\cdots,x_n\}=t_0\wedge s\leq t'$
  \item L6. $s\leq u\wedge t\leq u\Rightarrow s\leq t\vee t\leq s$
  \item L7. $s$ in $t=(\exists u,v\bullet t=u\smallfrown s\smallfrown v)$
  \item L8. $(\langle\{x_1,\cdots,x_n\}\rangle\smallfrown s)$ in $t\equiv t\neq\langle\rangle\wedge((t_0=\{x_1,\cdots,x_n\}\wedge s\leq t')\vee(\langle\{x_1,\cdots,x_n\}\rangle\smallfrown s$ in $t'))$
  \item L9. $s\leq t\Rightarrow (s\upharpoonright A)\leq(t\upharpoonright A)$
  \item L10. $t\leq u\Rightarrow(s\smallfrown t)\leq(s\smallfrown u)$
\end{itemize}
\end{proposition}

\begin{proof}
These laws can be proven straightforwardly from the related definition of the ordering relation.
\end{proof}

\begin{definition}[Length]
The length of a trace $s$ is defined as the number of commas plus 1, and denoted $\sharp s$.
\end{definition}

\begin{proposition}[Laws of Length]
The laws of length for concurrent traces are as follows.
\begin{itemize}
  \item L1. $\sharp\langle\rangle = 0$
  \item L2. $\sharp\langle\{x_1,\cdots,x_n\}\rangle = 1$
  \item L3. $\sharp(s\smallfrown t)=(\sharp s)+ (\sharp t)$
  \item L4. $\sharp(t\upharpoonright(A\cup B))=\sharp(t\upharpoonright A)+ \sharp(t\upharpoonright B)-\sharp(t\upharpoonright(A\cap B))$
  \item L5. $s\leq t\Rightarrow \sharp s\leq \sharp t$
  \item L6. $\sharp(t^n)= n\times(\sharp t)$
\end{itemize}
\end{proposition}

\begin{proof}
These laws can be proven straightforwardly from the related definition of the length.
\end{proof}

The concurrent traces can also be implemented by LISP as traces in CSP.

\begin{definition}[Concurrent Traces of a Process]
The complete set of all possible concurrent traces of a process $P$ can be known in advance, and is denoted by a function $traces(P)$.
\end{definition}

\begin{proposition}[Laws of concurrent Traces of a Process]
The laws of concurrent traces of a process are as follows.
\begin{itemize}
  \item L1. $traces(STOP)=\{t|t=\langle\rangle\}=\{\langle\rangle\}$
  \item L2. $traces(\{c_1,\cdots,c_n\}\rightarrow P)=\{t|t=\langle\rangle \vee (t_0=\{c_1,\cdots,c_n\}\wedge t'\in traces(P))\}=\{\langle\rangle\}\cup\{\langle \{c_1,\cdots,c_n\}\rangle\smallfrown t|t\in traces(P)\}$
  \item L3. $traces(\{c_1,\cdots,c_n\}\rightarrow P\mid \{d_1,\cdots,d_m\}\rightarrow Q )=\{t|t=\langle\rangle \vee (t_0=\{c_1,\cdots,c_n\}\wedge t'\in traces(P)) \vee(t_0=\{d_1,\cdots,d_m\}\wedge t'\in traces(Q))\}$
  \item L4. $traces(\{x_1,\cdots,x_n\}:B\rightarrow P(x_1,\cdots,x_n))=\{t|t=\langle \rangle \vee (t_0\in B\wedge t'\in traces(P(t_0)))\}$
  \item L5. $traces(\mu X:A\bullet F(X))=\bigcup_{n\geq 0}traces(F^n(STOP_A))$
  \item L6. $\langle\rangle\in traces(P)$
  \item L7. $s\smallfrown t\in traces(P)\Rightarrow s\in traces(P)$
  \item L8. $traces(P)\subseteq (\alpha P)^\ast$
\end{itemize}
\end{proposition}

\begin{proof}
These laws can be proven straightforwardly from the related definition of concurrent traces of a process.
\end{proof}

The concurrent traces of a process can also be implemented by LISP.

\begin{definition}[After]
If $s\in traces(P)$, $P$ after $s$ denoted $P/s$ is behaves the same as $P$ after $P$ has executed all events in $s$.
\end{definition}

\begin{proposition}[Laws of After]
The laws of after for concurrent traces are as follows.
\begin{itemize}
  \item L1. $P/\langle\rangle = P$
  \item L2. $P/(s\smallfrown t)=(P/s)/t$
  \item L3. $(\{x_1,\cdots,x_n\}:B\rightarrow P(x_1,\cdots,x_n))/\langle\{c_1,\cdots, c_n\rangle=P(c_1,\cdots,c_n)\}$, if $c_1,\cdots,c_n\in B$
      \begin{itemize}
        \item L3A. $(\{c_1,\cdots,c_n\}\rightarrow P)/\langle\{c_1,\cdots,c_n \}\rangle = P$
      \end{itemize}
  \item $traces(P/\langle s\rangle)=\{t|s\smallfrown t\in traces(P)\}$, if $s\in traces(P)$
\end{itemize}
\end{proposition}

\begin{proof}
These laws can be proven straightforwardly from the related definition of after.
\end{proof}

\begin{definition}[Change of Symbol]
The change of symbol for concurrent traces is a function $f^\ast:A^\ast\rightarrow B^\ast$ with $f:A\rightarrow B$.
\end{definition}

\begin{proposition}[Laws of Change of Symbol]
The laws of change of symbol for concurrent traces are the same as for traces in CSP, we retype them as follows.
\begin{itemize}
  \item L1. $f^\ast(\langle\rangle)=\langle\rangle$
  \item L2. $f^\ast(\langle \{x_1,\cdots,x_n\}\rangle)=\langle f(x_1,\cdots,x_n)\rangle$
  \item L3. $f^\ast(s\smallfrown t)=f^\ast(s)\smallfrown f^\ast(t)$
  \item L4. $f^\ast(s)_0=f(s_0)$, if $s\neq\langle\rangle$
  \item L5. $\sharp f^\ast(s)=\sharp s$
  \item L6. $f^\ast(s\upharpoonright A)=f^\ast(s)\upharpoonright f(A)$, if $f$ is an injection.
\end{itemize}
\end{proposition}

\begin{proof}
These laws can be proven straightforwardly from the related definition of change of symbol.
\end{proof}

\begin{definition}[Another Catenation]
If $s$ is a sequence and all its elements are also sequences, another catenation $\smallfrown/s$ is defined as catenating all the elements together in the original order.
\end{definition}

\begin{proposition}[Laws of Another Catenation]
The laws of another catenation for concurrent traces are the same as for traces in CSP, we retype them as follows.
\begin{itemize}
  \item L1. $\smallfrown/\langle\rangle=\langle\rangle$
  \item L2. $\smallfrown/\langle s\rangle=s$
  \item L3. $\smallfrown/(s\smallfrown t)=(\smallfrown/s)\smallfrown(\smallfrown/t)$
\end{itemize}
\end{proposition}

\begin{proof}
These laws can be proven straightforwardly from the related definition of another catenation.
\end{proof}

\begin{definition}[Concurrent Composition]
A sequence $s$ is a concurrent composition ($\textsf{cc}$) of two sequences $t$ and $u$, if it can be composed one by one from $t$ and $u$.
\end{definition}

\begin{proposition}[Laws of Concurrent Composition]
The laws of concurrent composition for concurrent traces are as follows.
\begin{itemize}
  \item L1. $\langle\rangle \textsf{ cc } s\equiv s$
  \item L2. $s\textsf{ cc } \langle\rangle\equiv s$
  \item L3. $(\langle x \rangle)\smallfrown s \textsf{ cc } \langle y\rangle\smallfrown t\equiv (\langle\{x,y\}\rangle\smallfrown(s\textsf{ cc } t))$
\end{itemize}
\end{proposition}

\begin{proof}
These laws can be proven straightforwardly from the related definition of concurrent composition.
\end{proof}

\begin{definition}[Subscription]
The $i^{th}$ element of the sequence $s$ denoted $s[i]$ for $=\leq i\leq \sharp s$.
\end{definition}

\begin{proposition}[Laws of Subscription]
The laws of subscription for concurrent traces are the same as for traces in CSP, we retype them as follows.
\begin{itemize}
  \item L1. $s[0]=s_0\wedge s[i+1]=s'[i]$, if $s\neq\langle\rangle$
  \item L2. $(f^\ast(s))[i]=f(s[i])$, for $i<\sharp s$
\end{itemize}
\end{proposition}

\begin{proof}
These laws can be proven straightforwardly from the related definition of subscription.
\end{proof}

\begin{definition}[Reversal]
The reversal of a sequence $s$ denoted $\overline{s}$, is obtained by taking its elements in reverse order.
\end{definition}

\begin{proposition}[Laws of Reversal]
The laws of reversal for concurrent traces are as follows.
\begin{itemize}
  \item L1. $\overline{\langle\rangle}=\langle\rangle$
  \item L2. $\overline{\langle\{x_1,\cdots,x_n\}\rangle}=\langle\{x_1,\cdots,x_n\}\rangle$
  \item L3. $\overline{s\smallfrown t}=\overline{t}\smallfrown\overline{s}$
  \item L4. $\overline{\overline{s}}=s$
  \item L5. $\overline{s}[i]=s[\sharp s-i-1]$, for $i\leq\sharp s$
\end{itemize}
\end{proposition}

\begin{proof}
These laws can be proven straightforwardly from the related definition of reversal.
\end{proof}

\begin{definition}[Selection]
The selection $s\downarrow \{x_1,\cdots, x_n\}$ of a sequence of pairs $s$ is obtained by replacing each pair by its second element, if its first element is $\{x_1,\cdots, x_n\}$. \end{definition}

\begin{proposition}[Laws of Selection]
The laws of selection for concurrent traces are as follows.
\begin{itemize}
  \item L1. $\langle\rangle\downarrow \{x_1,\cdots,x_n\}=\langle\rangle$
  \item L2. $(\langle\{y_1,\cdots,y_m\}.\{z_1,\cdots,z_k\}\rangle\smallfrown t\downarrow\{x_1,\cdots,x_n \})=t\downarrow\{x_1,\cdots,x_n\}$, if $m\neq n$ or $y_i\neq x_i(1\leq i\leq m=n)$
  \item L3. $(\langle\{x_1,\cdots,x_n\}.\{z_1,\cdots,z_k\}\rangle\smallfrown t)\downarrow\{x_1,\cdots,x_n\}=\langle\{z_1,\cdots,z_k\}\rangle\smallfrown(t\downarrow \{x_1,\cdots,x_n\})$
\end{itemize}
\end{proposition}

\begin{proof}
These laws can be proven straightforwardly from the related definition of selection.
\end{proof}

\begin{definition}[Composition]
The composition of sequences $s$ and $t$ denoted $(s;t)$, means that when $s$ is successfully terminated (denoted $\surd$), $s$ starts.
\end{definition}

\begin{proposition}[Laws of Composition]
The laws of composition for concurrent traces are as follows.
\begin{itemize}
  \item L1. $s;t=s$, if $\neg(\langle\surd\rangle$ in $s)$
  \item L2. $(s\smallfrown\langle\surd\rangle);t=s\smallfrown t$, if $\neg(\langle\surd\rangle$ in $s)$
      \begin{itemize}
        \item L2A. $(s\smallfrown\langle\surd\rangle\smallfrown u);t=s\smallfrown t$, if $\neg(\langle\surd\rangle$ in $s)$
      \end{itemize}
  \item L3. $s;(t;u)=(s;t);u$
  \item L4A. $s\leq t\Rightarrow ((u;s)\leq(u;t))$
  \item L4B. $s\leq t\Rightarrow ((s;u)\leq(t;u))$
  \item L5. $\langle\rangle;t=\langle\rangle$
  \item L6. $\langle\surd\rangle;t=t$
  \item L7. $s;\langle\surd\rangle=s$, if $\neg(\langle\surd\rangle$ in $(\overline{s})')$
\end{itemize}
\end{proposition}

\begin{proof}
These laws can be proven straightforwardly from the related definition of composition.
\end{proof}

\begin{definition}[Specification]
A specification of a product $P$ is a description of the way it is intended to behave, which contains free variables standing some observable behaviors. A specification is denoted by $\mathcal{S}(tr)$, where $tr$ are free variables.
\end{definition}

\begin{definition}[Satisfaction]
A product $P$ satisfies a specification $\mathcal{S}$ when $P$ meets $\mathcal{S}$, denoted by $P \textsf{ sat }\mathcal{S}$.
\end{definition}

\begin{proposition}[Laws of Specification and Satisfaction]
The laws of specification and satisfaction for concurrent traces are the same as for traces in CSP, and we retype them as follows.
\begin{itemize}
  \item L1. $P\textsf{ sat }true$
  \item L2A. If $P\textsf{ sat }\mathcal{S}$ and $P\textsf{ sat }\mathcal{T}$, then $P\textsf{ sat }(\mathcal{S}\wedge\mathcal{T})$
  \item L2. If $\forall n\bullet(P\textsf{ sat }\mathcal{S}(n))$, then $P\textsf{ sat }(\forall n\bullet\mathcal{S}(n))$, if $P$ has not relation to $n$
  \item L3. If $P\textsf{ sat }\mathcal{S}$ and $\mathcal{S}\Rightarrow\mathcal{T}$, then $P\textsf{ sat }\mathcal{T}$
  \item L4A. $STOP\textsf{ sat }(tr=\langle\rangle)$
  \item L4B. If $P\textsf{ sat }\mathcal{S}(tr)$, then $(\{c_1,\cdots, c_n\}\rightarrow P)\textsf{ sat }$ $(tr=\langle\rangle\vee (tr_0=\{c_1,\cdots,c_n\}\wedge\mathcal{S} (tr')))$
  \item L4C. If $P\textsf{ sat }\mathcal{S}(tr)$, then $(\{c_1,\cdots,c_n\}\rightarrow \{d_1,\cdots,d_m\}\rightarrow P)\textsf{ sat }$ $(tr\leq \langle \{c_1,\cdots,c_n\} ,\{d_1,\cdots,d_m\}\rangle\vee (tr\geq\langle \{c_1,\cdots,c_n\},\{d_1,\cdots,d_m\}\rangle\wedge\mathcal{S}(tr'')))$
  \item L4D. If $P\textsf{ sat }\mathcal{S}(tr)$ and $Q\textsf{ sat }\mathcal{T}(tr)$, then $(\{c_1,\cdots,c_n\}\rightarrow P\mid\{d_1,\cdots,d_m\}\rightarrow Q)\textsf{ sat }$ $(tr=\langle\rangle \vee (tr_0=\{c_1,\cdots,c_n\}\wedge \mathcal{S}(tr')) \vee (tr_0=\{d_1,\cdots,d_m\}\wedge\mathcal{T}(tr')))$
  \item L4. If $\forall \{x_1,\cdots,x_n\}:B\bullet(P(x_1,\cdots,x_n)\textsf{ sat }\mathcal{S}(tr,x_1,\cdots,x_n))$, then $(\{x_1,\cdots,x_n\}:B\rightarrow P(x_1,\cdots,x_n))\textsf{ sat }$ $(tr=\langle\rangle \vee(tr_0\in B\wedge \mathcal{S}(tr',tr_0)))$
  \item L5. If $P\textsf{ sat }\mathcal{S}(tr)$ and $s\in traces(P)$, then $P/s\textsf{ sat }\mathcal{S}(s\smallfrown tr)$
  \item L6. If $F(X)$ is guarded and $STOP\textsf{ sat }\mathcal{S}$ and $((X\textsf{ sat }\mathcal{S})\Rightarrow(F(X)\textsf{ sat }\mathcal{S}))$, then $\mu X\bullet F(X)\textsf{ sat }\mathcal{S}$
\end{itemize}
\end{proposition}

\begin{proof}
These laws can be proven straightforwardly from the related definitions of specification and satisfaction.
\end{proof}

\section{Concurrency}\label{concurrency}

In this section, we will inspect concurrency for CCP, including interaction and concurrency.

\subsection{Interaction}

\begin{definition}[Interaction]
Two processes interact with each other, and with the same alphabet, denoted $P\parallel Q$.
\end{definition}

\begin{proposition}[Laws of Interaction]
The laws of interaction of CCP are almost the same as those of CSP, we retype them as follows.

\begin{itemize}
  \item L1. $P\parallel Q = Q\parallel P$
  \item L2. $P\parallel (Q\parallel R)=(P\parallel Q)\parallel R$
  \item L3A. $P\parallel STOP_{\alpha P}=STOP_{\alpha P}$
  \item L3B. $P\parallel RUN_{\alpha P}=P$
  \item L4A. $(c\rightarrow P)\parallel (c\rightarrow Q)=(c\rightarrow(P\parallel Q))$
  \item L4B. $(c\rightarrow P)\parallel (d\rightarrow Q)=STOP$, if $c\neq d$
  \item L4. $(x:A\rightarrow P(x))\parallel (y:B\rightarrow Q(y))=(z:(A\cap B)\rightarrow (P(z)\parallel Q(z)))$
\end{itemize}
\end{proposition}

\begin{proof}
These laws can be proven straightforwardly from the related definition of interaction.
\end{proof}

For the implementation of interaction operator $\parallel$ by LISP, we omit it.

\begin{proposition}[Laws of Traces of Interaction Operator]
The laws of traces of interaction operator are as follows.
\begin{itemize}
  \item L1. $traces(P\parallel Q)=traces(P)\cap traces(Q)$
  \item L2. $(P\parallel Q)/s=(P/s)\parallel(Q/s)$
\end{itemize}
\end{proposition}

\begin{proof}
These laws can be proven straightforwardly from the related definitions of interaction and concurrent traces.
\end{proof}

\subsection{Concurrency}

\begin{definition}[Concurrency]
Concurrency is a generalization of the case of interaction, while $P$ and $Q$ may have different alphabets, and it is also denoted by $P\parallel Q$.
\end{definition}

\begin{proposition}[Laws of Concurrency]
The laws of concurrency of CCP are quite different to those of CSP, we give them as follows.

\begin{itemize}
  \item L1. $P\parallel Q = Q\parallel P$
  \item L2. $P\parallel (Q\parallel R)=(P\parallel Q)\parallel R$
  \item L3A. $P\parallel STOP_{\alpha P}=STOP_{\alpha P}$
  \item L3B. $P\parallel RUN_{\alpha P}=P$
  \item L4. $(c\rightarrow P)\parallel (c\rightarrow Q)=(c\rightarrow(P\parallel Q))$
  \item L5. $(c\rightarrow P)\parallel (d\rightarrow Q)=(\{c,d\}\rightarrow(P\parallel Q))$, if $c\neq d$
  \item L6. $(x:A\rightarrow P(x))\parallel (y:B\rightarrow Q(y))=(z:(A\cup B)\rightarrow (P\parallel Q))$, $z=x=y$, otherwise, z=\{x,y\}
\end{itemize}
\end{proposition}

\begin{proof}
These laws can be proven straightforwardly from the related definition of concurrency.
\end{proof}

The concurrency operator also can be implemented by LISP, we omit the implementation.

\begin{proposition}[Laws of Traces of Concurrency Operator]
The laws of traces of concurrency operator are as follows.
\begin{itemize}
  \item L1. $traces(P\parallel Q)=traces(P)\cap traces(Q)$, if $\alpha P=\alpha Q$
  \item L2. $traces(P\parallel Q)=traces(P)\textsf{ cc } traces(Q)$, if $\alpha P\cap\alpha Q=\{\}$
  \item L3. $(P\parallel Q)/s=(P/(s\upharpoonright\alpha P))\parallel(Q/(s \upharpoonright\alpha Q))$
  \item L4. $traces(P\parallel Q)=\{t|t\in(\alpha P\cup\alpha Q)^\ast\}$
\end{itemize}
\end{proposition}

\begin{proof}
These laws can be proven straightforwardly from the related definitions of concurrency and concurrent traces.
\end{proof}

\subsection{Specification}

Let $P$ and $Q$ be processes, $tr\in traces(P\parallel Q)$, and $P\textsf{ sat }\mathcal{S}(tr)$, $Q\textsf{ sat }\mathcal{T}(tr)$. So $(tr\upharpoonright\alpha P)\in traces(P)$ 
and $P\textsf{ sat }\mathcal{S}(tr\upharpoonright\alpha P)$, and also $tr\upharpoonright\alpha Q\in traces(Q)$ and $Q\textsf{ sat }\mathcal{T}(tr\upharpoonright\alpha Q)$.

\begin{proposition}[Laws of Specification for Concurrency]
The laws of specification for concurrency operator are as follows.
\begin{itemize}
  \item L1. If $P\textsf{ sat }\mathcal{S}(tr)$, and $Q\textsf{ sat }\mathcal{T}(tr)$, then 
  $(P\parallel Q)\textsf{ sat }(\mathcal{S}(tr\upharpoonright\alpha P)\wedge \mathcal{T}(tr\upharpoonright\alpha Q))$
  \item L2. If $P$ and $Q$ never stop, then $P\parallel Q$ never stops.
\end{itemize}
\end{proposition}

\begin{proof}
These laws can be proven straightforwardly from the related definitions of concurrency and specification.
\end{proof}

\subsection{Mathematical Theory of Deterministic Processes}

The laws in the above two sections are true in intuition, to make them strict, it is needed to give those denotational semantics. 
Fortunately, the denotational semantics of each operator of deterministic processes are the same as those in CSP, including the interaction operator and concurrency operator. 
And the proofs of the existence of solution for recursive equation in the law Proposition \ref{LoP}.L2 and the uniqueness of the solution by use of Scott's fix-point theory, 
are also the same as those of CSP, we omit them (please refer to CSP \cite{CSP}).

\section{Other Contents of CCP}\label{occ}

We inspect the rest of CCP corresponding to CSP, including nondeterminism, communication, etc, and we find that the rest of CCP is almost the same as the corresponding parts of CSP. So, there is not necessary to retype the quite long contents again. Indeed, the main difference between CCP and CSP is the treatment of concurrency.

\subsection{Nondeterminism}

The nondeterminism is defined by internal choice operator, general choice operator, refusals, concealment, but no interleaving operator. The nondeterminism of CCP are almost the same as that of CSP, we omit it.

\subsection{Communication}

The communication is defined by input and output, the communication of CCP is also almost the same as CSP, we also omit it.

\section{Conclusions}\label{con}

We modify CSP to CCP, and the main difference between the two algebras is the treatment of concurrency. CSP only permits a process to engage in one event on a moment and record only this one event in the process traces, while CCP permits a process to engage in multi events simultaneously and record simultaneously these events in the process traces.

\label{lastpage}

\end{document}